\documentclass[runningheads]{llncs}
\usepackage[utf8]{inputenc}

\usepackage[ruled,noline,noend,linesnumbered]{algorithm2e}
\usepackage{amsmath}
\usepackage[textsize=small]{todonotes}
\presetkeys%
    {todonotes}%
    {inline}{}
\usepackage{amssymb}
\usepackage{graphicx}
\usepackage{tabularx}
\usepackage{hyperref}
\usepackage{xcolor}

\usepackage[maxbibnames=99]{biblatex}
\newcommand{\forkjoin}{fork{\text -}join}
\newcommand{\yo}{\gamma^{out}}
\newcommand{\yi}{\gamma^{in}}
\newcommand{\mO}{\mathcal{O}}
\newcommand{\Cmax}{C_{max}}
\newcommand{\p}{p}

\newcommand{\n}{n}
\newcommand{\M}{M}
\newcommand{\m}{m}

\title{Scheduling Fork-Join Task Graphs to Heterogeneous Processors}

\author{Huijun Wang and Oliver Sinnen} 

\institute{Department of Electrical, Computer, and Software Engineering \\
University of Auckland, New Zealand}

\begin{document}

\maketitle


\begin{abstract}
The scheduling of task graphs with communication delays has been extensively studied. Recently, new results for the common sub-case of fork-join shaped task graphs were published, including an EPTAS and polynomial algorithms for special cases. These new results modelled the target architecture to consist of homogeneous processors. However, forms of heterogeneity become more and more common in contemporary parallel systems, such as CPU--accelerator systems, with their two types of resources.
In this work, we study the scheduling of fork-join task graphs with communication delays, which is representative of highly parallel workloads, onto heterogeneous systems of related processors. We present an EPAS, and some polynomial time algorithms for special cases, such as with equal processing costs or unlimited resources. Lastly, we briefly look at the above described case of two resource-types and its implications. It is interesting to note, that all results here also apply to scheduling independent tasks with release times and deadlines.
\end{abstract}

\section{Introduction}

Task scheduling is one of the steps in parallelising a computational workload, where units of work (called tasks) are scheduled onto the available resources, to optimise for the schedule length (makespan). The structure of such a workload can be represented by a directed acyclic graph, where nodes represent tasks, edges represent data dependencies between them, and weights represent the associated costs of computation (processing) and communication (data transfer). This problem is generally intractable and difficult to approximate as well. 
In a typically studied model, tasks are scheduled across homogeneous processors, with communications between each other, each incurring communication delays. With varied and evolving architectures, it is becoming more appropriate to model systems of heterogeneous processors. Although results are generalised to an arbitrary number of processors, a typical example of heterogeneity consists of two types of resources, each representing a general purpose processor or one type of accelerator resource such as GPU, FPGA, or other specialised hardware, and sometimes with clusters of each kind in proximity. Many high performance parallel computing systems, in particular the largest supercomputers \cite{top500}, have such an architecture. 
In terms of the workload structure, we use the fork-join structure in particular, which represents a significant subclass of parallel computations, as well as being a basic substructure that is common to other structures such as the series-parallel graph. It consists of a source task forking into many branch tasks, which in turn join back into a sink task. The source task prepares and divides the branch tasks, and the sink task gathers and combines results for subsequent steps. This can represent master-worker types of computations etc., and is representative of highly parallel workloads in general. More examples include MapReduce computation with frameworks like Hadoop or scatter/gather communication when using MPI parallelism \cite{MPI}. Also the parallel Executor-Service in Java supports such a fork-join pattern \cite{oracle}. Even if these applications may need to schedule dynamically, or have other nuances, the theoretic treatment here is important to guiding heuristic design, creating benchmarks, and adds to the completeness of theoretical results in classical scheduling.


Workloads that make use of accelerators are often highly parallel, and lend themselves to the fork-join structure. This paper investigates the optimisation problems that involve scheduling fork-join structures to heterogeneous processors, and describes some theoretical solutions to them, which have been extended from scheduling to homogeneous processors. In all cases, the algorithms also apply to the equivalent problems in scheduling independent tasks with release times and deadlines. 


The rest of the paper is organised as follows. Section \ref{sec:problem} describes the problem formally and summarises our results. Section \ref{sec:related} outlines related works. Section \ref{sec:eptas} describes an EPAS parameterised by the ratio of processor speeds and section \ref{sec:exact} describes some exact algorithms for special cases. 


\section{Problem Definition} \label{sec:problem}

In this section we define the problem, including models of the processing systems as well as our workload. We begin with the task graph and the fork-join structure.

\subsection{Task Graph}
A task graph is a directed acyclic graph with nodes representing tasks and edges representing precedence constraints (which are communications). 

\begin{figure}[ht] \label{DAG}
\centering
\includegraphics[width=0.7\textwidth]{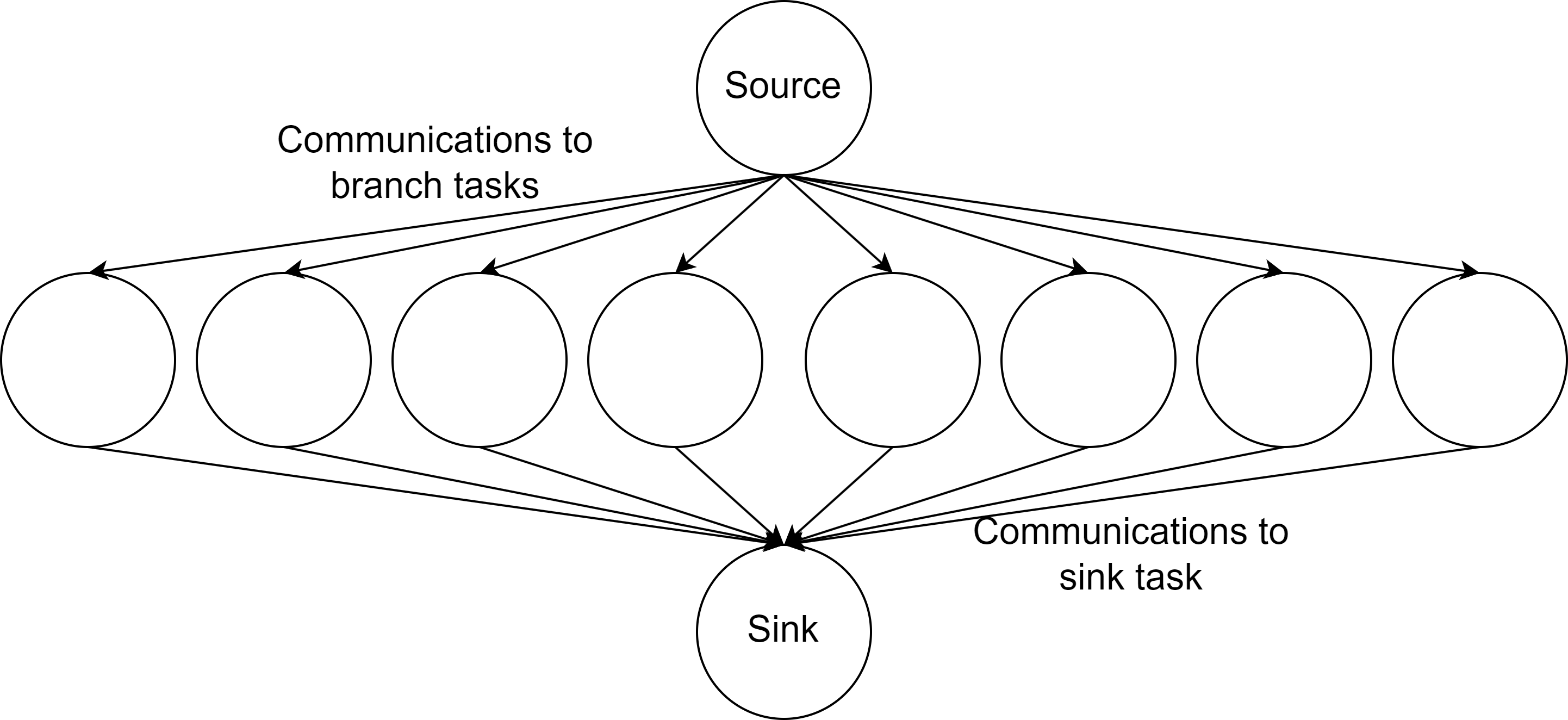}
\caption{Example of the parallel fork-join structure}
\end{figure}

The shape of a fork-join task graph is shown figure \ref{DAG}. It has a set of tasks $\boldsymbol{J}$, including the source task $\boldsymbol{j_{src}}$ and sink task $\boldsymbol{j_{sink}}$. Each task $\boldsymbol{j \in J}$ has an associated processing cost $\boldsymbol{\p_j \in P}$, and each branch task $j \in J \setminus \{j_{src}, j_{sink}\}$ has an incoming communication from the source, with cost $\boldsymbol{\yi_j \in \Gamma}$, and outgoing communication to the sink, with cost $\boldsymbol{\yo_j \in \Gamma}$. It is to be scheduled on the following system.

\subsection{System Model}

In systems of uniformly related heterogeneous processors, each processor $\boldsymbol{\m \in \M}$ has a speed $\boldsymbol{s_\m \in S}$, and executes a task $j$ in time $\p_j/ s_\m $. 

Communication times are simply equal to the communication costs (there is no heterogeneity involved with the processing of communications), and no communication costs are incurred between tasks scheduled to the same processor (this is a common assumption in other works, e.g. \cite{beaumont2020scheduling,Drozdowski2009:spp}). Systems are modelled this way due to slower communications between processing resources not sharing memory. 

Tasks scheduled to the same processor as the source or sink are said to be \emph{local}, and tasks scheduled to other processors are said to be \emph{remote}.


\subsection{Schedule}

A schedule consists of processor allocations $\boldsymbol{\m_j \in \M}$ and start time allocations $\boldsymbol{\sigma_j \in \mathbb{N}_0}$ for all tasks $j \in J$. The start time allocations are equivalent to a total order ($\leq_\sigma$) of execution over each set of tasks scheduled to the same processor, because optimally, tasks start as early as possible given their order. 
Therefore, for $i,j \in J$ scheduled consecutively on the same processor
(meaning $i\leq_\sigma j$ and $\nexists \, k \in J\colon i\leq_\sigma k \leq_\sigma j$), the start time of $j$ is 

\[
\sigma_j=
\begin{cases}
\max(\sigma_i+\p_i / s_{m_{i}},\, \p_{src} / s_{m_{src}} + \yi_j), & \text{if } \m_j\neq\m_{src}\\
\sigma_i+\p_i / s_{m_{i}}, & \text{if } \m_j=\m_{src}
\end{cases}
\]

For $j = \min\leq_\sigma$, the first branch task on a processor,
\[
\sigma_j=
\begin{cases}
\p_{src} / s_{m_{src}} + \yi_j, & \text{if } \m_j\neq\m_{src}\\
\p_{src} / s_{m_{src}}, & \text{if } \m_j=\m_{src}
\end{cases}
\]

The objective is to find a schedule with optimal schedule length.

While with homogeneous processors we could have ignored the processing times of the source and sink tasks without loss of generality, this simplification does not hold when processor systems are heterogeneous. With the execution times of the source and sink tasks taken into account, the schedule length to be minimised is:

\[
\max_{j\in J}\left(\sigma_j + \p_j / s_{m_{j}} + 
\begin{cases}
\yo_j, & \text{if } \m_j\neq\m_{sink}\\
0, & \text{if } \m_j=\m_{sink}
\end{cases}  \:\right) + \p_{sink} / s_{m_{sink}}
\]



\subsection{Scheduling with Release Times and Deadlines}
\label{sec:releaseTimes}
Scheduling fork-joins is related to scheduling independent tasks with release times and deadlines, where each task $j$ has an associated release time $\boldsymbol{r_j}$ and a deadline $\boldsymbol{d_j}$. When treated as a feasibility problem of scheduling under time $T$, a release time $r_j$ would be equivalent to the incoming communication delay added to the execution time of the source $\p_{src}/s_{m_{src}} + \yi_j$, and a deadline $d_j$ would be equivalent to $T - \p_{sink}/s_{m_{sink}} - \yo_j$. However, fork-join scheduling has to additionally consider the source and sink tasks and zeroed communication costs when tasks are scheduled to the same processor as the source or sink. Nevertheless, the results of this work for fork-joins, in particular the scheduling of the tasks allocated on remote processors, will also apply to scheduling with release times and deadlines. When extended to a minimisation problem, a minimisation objective when working with deadlines is $L_{max}$, the maximum lateness.

\subsection{Source and Sink Allocation for Fork-Join Task Graphs}

With the processing times of the source and sink tasks being different on each processor, they cannot be simplified away as with identical homogeneous processors, where the only choice that matters is whether the source and sink are on the same processor. Here, every different choice of processor type for the source and sink will affect the schedule. 
However, in the worst case, in a fully heterogeneous system and without any further analytical insight, an algorithm can be repeated $|S|^2$ times for each choice of source and sink processors, which would not change the class of complexity, or algorithm (for a theoretical purpose), summarised as the following.

\begin{proposition} \label{procs}
    If a fork-join scheduling problem $A$ can be solved in $sched(A)$ time with source and sink allocations given, then $A$ can be solved in $|S|^2 sched(A)$ time at worst.
\end{proposition}

Using the $\alpha|\beta|\gamma$ notation from \cite{GRAHAM1979287}, we have the following table to summarise our results.

\begin{table}[ht]
    \centering 
    \caption{Summary of Results} \label{tab:summary}
    \medskip
    \begin{tabular}{|l|l|l|}
    Problem & \hfill Result & \hspace{10pt} Section \\ \hline \hline
    $Qm|fork{\text -}join, c_{ij}|C_{max}$ && \\
    $Qm|r_j|L_{max}$ & \hfill \textcolor{gray}{sNP-hard}, EPAS &\hfill  \ref{sec:eptas}\\ \hline
    
    $Q|fork{\text -}join, c_{ij}, p_j=p|C_{max}$ & \hfill \textcolor{gray}{Open}, & \\
    $Q|r_j, p_j=p|L_{max}$& \hfill ($OPT{\text +}p_{max}$)-approx. &\hfill  \ref{sec:bipartite} \\ \hline
    
    $Q2|fork{\text -}join, c_{ij}, p_j=p|C_{max}$ && \\
    $Q2|r_j, p_j=p|L_{max}$ & \hfill P &\hfill  \ref{sec:2proc}\\ \hline
    
    $Q\infty|\forkjoin, c_{ij}, p_j=p|C_{max}$ && \\
    $Q\infty|r_j, p_j=p|L_{max}$ & \hfill P &\hfill  \ref{sec:UnlimitedProc}\\ \hline 
    
    \end{tabular}
\end{table}

\section{Related work} \label{sec:related} 

Heterogeneous systems have been studied as early as \cite{GareyJohnson78}. Recently, other theoretical results have been presented for scheduling to heterogeneous systems regarding its complexity \cite{KnopKoutecky:17:Scheduling-Meets}, as well as an EPTAS for scheduling independent tasks to uniformly related processors \cite{10.1007/978-3-642-02927-1_47}. There has existed a 2-approximation for scheduling to unrelated processors \cite{LenstraTardos:90:Approximation-algorithms}, and a PTAS for minimising the weighted finish times on related processors \cite{chekuri2001ptas}. 

Recent attention has also been on the kind of system with two resource types in particular (modelling the CPU--accelerator set-up), with both theoretical \cite{10.1007/978-3-030-29400-7_9} and practical results \cite{ait2020efficient} presented, and there has also been an entire review on scheduling to this system of two resource types \cite{beaumont2020scheduling}.

Regarding the scheduling of fork-join structures, recent work has been presented on an EPTAS \cite{jansen2021eptas} and other algorithms \cite{fork-join} for scheduling to homogeneous systems.




\section{Efficient Parameterised Approximation Scheme} \label{sec:eptas}

We describe an Efficient Parameterised Approximation Scheme (EPAS) for heterogeneous systems of uniformly related processors, as extended from an EPTAS for the homogeneous model \cite{jansen2021eptas}. The EPAS is parameterised by the ratio between the fastest and slowest processor speeds $\frac{s_{max}}{s_{min}}$ and returns a solution within $1 + \mO(\epsilon)$ of the optimal in $\mO(f(\frac{1}{\epsilon}, \frac{s_{max}}{s_{min}}) \times poly(|J|))$ time, where $|J|$ is the number of tasks. 

The basic procedure is outlined as follows. Given a feasibility problem of scheduling under makespan $T$, we simplify the instance and set up an ILP to decide if a schedule exists for the simplified instance - the simplifications maintain a guarantee that an optimum for the simplified instance is within a range of that to the original. 

The simplifications serve to bound the number of task sizes etc., such that there is a limited number of possible schedules (configurations) on each processor. The accuracy parameter $\epsilon$ decides the resolution of the simplified instance and is inversely related to the runtime of the approximation scheme. In summary, processing and communication times are truncated, processing times too small to truncate are represented by placeholders, and sizes of gaps limited. Here they are in more detail: 

\begin{itemize}
    \item Truncate communication times to the nearest multiple of $\epsilon T$.
    \item Distinguish between big and small tasks by processing time, with a cutoff of $\epsilon^2 T$, which will become the minimum task size (smaller tasks will be "bundled" together).
    \item Truncate the processing times of big tasks into the nearest order of magnitude:
    \begin{equation*} \label{eq:bigTaskSize}
    \{(1+\epsilon)^n\epsilon^2 T \mid n\in \mathbb{N}_0\}
    \end{equation*}
    so that their errors are proportional to their sizes, and when added together the total error in a configuration is bounded and proportional to $T$ ($\epsilon T$). Now the tasks can be categorised into a finite number of sizes $(\p, \yi, \yo) \in P \times \Gamma^2$.
    \item Replace the total volume of small tasks with an equivalent number of identical placeholders of size $\epsilon^3 T$ - for each communication class. Only an error of $\epsilon^3 T$ is introduced for each of the $1/\epsilon^2$ communication classes, and in total $\epsilon T$.
    \item Bundle the placeholder tasks into blocks of $1/\epsilon$, with total processing time $\epsilon^2 T$ (with a limited number of stubs of $\leq 1/\epsilon$ placeholders being allowed).
    \item The last step limits the resolution of start times by limiting the sizes of gaps to an increment ($\epsilon^2 T$).
\end{itemize}

After the simplifications, there are $|\Gamma| = 1/\epsilon$ different communication times and $|P| = 1/\epsilon \log 1/\epsilon$ different processing times, i.e. a constant number in terms of the problem input size.

The ILP itself is what can be called a configuration ILP. We can describe a configuration as a (maximal) set of tasks that can be feasibly scheduled onto a single processor. More accurately, it is defined as a multiset of task sizes, representing the numbers of slots of each size (as created in the simplifications) that can be provided by a processor adopting such a configuration. The ILP selects a configuration for each processor, out of all possible configurations (which are bounded in number by the simplifications), such that all the configurations together provide enough slots to accommodate the given numbers of tasks of each size.


Let the set of all configurations be denoted by $\boldsymbol{C}$. Variables $x_C \ \forall C \in \boldsymbol{C}$ select the number of each configuration used, and each configuration $C$ provides a $C_{\p, \yi, \yo}$ number of slots for tasks with processing time $\p$ and communication times $\yi$ and $\yo$.





\subsection{EPAS for Related Processors $Qm|\forkjoin, c_{ij}|C_{max}$}

We now describe the formulation of the EPAS for scheduling the fork-join graph to a heterogeneous system of uniformly related processors, parameterised by the ratio in speed between the fastest and slowest processors, $\frac{s_{max}}{s_{min}}$.
Remember, with related heterogeneous processors, each processor $\m \in \M$ completes a task $j$ in $\p_j/ s_\m$ time, where $s_\m \in S$ is the speed of the processor. The input to the problem is then a set of tasks with given processing costs and communication times (where communication times are as given and do not depend on processing speed), and a set of processors with given speeds.

With this, we first redefine a \emph{small task} to be any task that takes less than $t_\textsc{small} = \epsilon^2 T$ to execute on the \emph{slowest} processor (and therefore has cost no greater than $s_{min}\epsilon^2 T$), and the rest are treated as \emph{big tasks}. This way, a small task will never become a "big task" (in the original sense that it takes more than $\epsilon^2 T$ time to execute) moving from one processor to another, while a big task may become a "small task" (taking less than $\epsilon^2 T$ to execute) depending on the processor. The latter has no impact on the formulation of the following ILP and does not impact the complexity. Small tasks will be replaced by one or more placeholders of size $\p_\textsc{small} = s_{min}\epsilon^3 T$, bounding the number of processing costs. Something to note is that instead of other representations of small tasks in PTAS for scheduling, placeholders are used here to manage communication costs, and although they are fine-grained in order to reduce rounding error, they are bundled together on processor configurations so as not to increase complexity.

With the maximum processing cost being $s_{max}T$, and $(1+\epsilon)^n s_{min}\epsilon^2 T \leq s_{max}T$, there are $n \leq \mO(\frac{1}{\epsilon} (\log \frac{s_{max}}{s_{min}} + \log\frac{1}{\epsilon}))$ cost sizes. On the other hand, the configurations only use $\mO(\frac{1}{\epsilon} \log\frac{1}{\epsilon})$ different rounded processing times (as with before). Let these processing times be denoted $t \in P'$.

The input to the new ILP will again consist of the numbers of tasks of each size (cost) category $(\p, \yi, \yo)$, where $\p$ here is the processing cost (instead of the actual processing times). This includes the numbers of placeholder tasks that will replace the combined volume of small tasks for each category of $\yi, \yo \in \Gamma^2$. Therefore we have the following values in the right hand side of the ILP, where $N_{\textsc{small}, \yi, \yo}$ is the number of placeholder tasks to replace the small tasks.

\begin{gather*}
N_{\p, \yi, \yo} \ \ \forall (\p, \yi, \yo) \in P \times \Gamma^2 \\
N_{\textsc{small}, \yi, \yo} \ \ \forall (\yi, \yo) \in P \times \Gamma^2 
\end{gather*}

We round each of the processing speeds $s_\m \in {S}$ down to the nearest $s_{min}(1+\epsilon)^n, n \in \mathbb{N}_0$, where $s_{max}$ is the highest, giving $\frac{ \log(s_{max}/s_{min}) }{ \log(1+\epsilon)}$ types of processors. 




Now, we introduce variables which select the numbers of tasks of each size (cost) to allocate to each type of processor. This includes large sizes and the placeholders for small tasks.
\begin{gather*}
\n_{\p, \yi, \yo}^{s} \ \forall s \in S \\
\n_{\textsc{small}, \yi, \yo}^{s} \ \forall s \in S 
\end{gather*}

The following constraints make sure that the total number of each task size allocated on all processors is at least as large as the given number of tasks from the input.


\begin{equation} \label{eq:first}
\sum_{s \in {S}} \n_{\p, \yi, \yo}^{s} \geq N_{\p, \yi, \yo}
\end{equation}
\begin{equation}
\sum_{s \in {S}} \n_{\textsc{small}, \yi, \yo}^{s} \geq N_{\p, \yi, \yo}
\end{equation}
\begin{equation*}
\forall (\p, \yi, \yo) \in P \times \Gamma^2
\end{equation*}

Let $(t, \yi, \yo)$ describe the size category of a task in terms of its actual execution time $t_{\p_j, s}$ on an assigned processor with speed $s$. For a target makespan $T$, the processing times of each task on each processor is put in terms of $T$, and from this it is known which size category $(t, \yi, \yo)$ the task belongs to on each processor. In the end, each processor will adopt some configuration which is chosen from the same set of configurations as in the original formulation for the homogeneous-processors case, with each configuration being some set of tasks that could feasibly be scheduled together on this processor, encoded in the ILP as a vector of multiplicities of the task sizes. 

Let the variables $x^s_C \ \forall s \in S, \forall C \in \boldsymbol{C}$ decide the number of processors of speed $s$ to use configuration $C$, and let $C_{t, \yi, \yo}$ denote the number of tasks of size $(t, \yi, \yo)$ that can go into the configuration $C$. Then, for each processor speed $s$ and task size $(t_{\p, s}, \yi, \yo)$, the total number of tasks of this size which can be accommodated by the selected configurations must be enough for the total number of tasks which will execute with this size on this type of processor.

\begin{equation}
\sum_{C\in\boldsymbol{C}} x^s_{C} C_{t, \yi, \yo} - \sum_{\p\in P} n^s_{\p, \yi, \yo} \cdot a_{\p/s=t} \geq 0
\end{equation}
\begin{equation*}
\forall s \in S, \ \forall t \in P', \ \forall(\yi, \yo) \in \Gamma^2 
\end{equation*}

where $a_{\p/s=t}$ is defined by

\begin{gather*}
a_{\p/s=t}=
\begin{cases}
1, & \text{\small if with speed $s$, task with cost $\p$ executes in time category $t$} \\
0, & \text{\small otherwise}
\end{cases} \\
\forall s \in S, \ \forall \p \in P, \ \forall t \in P'
\end{gather*}

The placeholders for small tasks are treated slightly differently. The different lengths of time it takes to execute a placeholder task on each type of processor is encoded into $a_{\textsc{small}/s} \ \forall s \in S$, where it is simply $a_{\textsc{small}/s} = \p_\textsc{small}/s$. Next, for each size category (which remain differentiated by communication times) $\yi, \yo$, let $C_{\textsc{small}, \yi, \yo} \ \forall C \in \boldsymbol{C}$ represent the available time provided by $C$ for placeholders of this category. Then, for each type of processor, the total time given to placeholders of each category ($\yi, \yo$) by the selected configurations must accommodate the total amount of time it will take to execute the placeholders assigned to it. In addition, some big tasks will take less than $t_\textsc{small}$ to execute on faster processors, where they can then be treated as small tasks and added to the total volume of small task execution time assigned to these processors. For such tasks, the factors $a_{\p/s} \ \forall s \forall \p$ will give the amount of time they will each add to the total volume of small tasks on processors of speed $s$. 

\begin{equation*}
a_{\p/s}=
\begin{cases}
\p/s, & \text{\small if } \p/s \leq t_\textsc{small} \\
0, & \text{\small otherwise}
\end{cases} \\
\end{equation*}
\begin{equation}
\sum_{C\in\boldsymbol{C}} x^s_{C} C_{\textsc{small}, \yi, \yo} 
- \sum_{\p\in P} n^s_{\p, \yi, \yo} \cdot a_{\p/s} - n^s_{\textsc{small}, \yi, \yo} \cdot a_{\textsc{small}/s}  \geq 0
\end{equation}
\begin{equation*}
\forall s \in S, \ \forall t \in P', \ \forall(\yi, \yo) \in \Gamma^2 
\end{equation*}



Finally, for each type of processor, a number of configurations is selected that is no more than the available processors of that type.

\begin{equation} \label{eq:last}
\sum_{C\in\boldsymbol{C}} x_C^s \leq |\M^s|, \ \forall s \in S
\end{equation}

In summary, the ILP is defined by constraints \ref{eq:first}-\ref{eq:last}, and has variables
$x^s_C \ \forall s \forall C$, $n^s_\p \ \forall s \forall \p$, and $n^s_\textsc{small} \ \forall s$. 

\begin{theorem} \label{epas}
There is an EPAS for scheduling fork-join structures to uniformly related heterogeneous systems that gives a schedule in 
\[ 2^{\mathcal{O}({1/\epsilon^4\log^2 1/\epsilon \log\frac{s_{max}}{s_{min}}})} \mathcal{O}(\log^2 N) + \mathcal{O}(N\log N) \]
time that is no worse than $\frac{(1+6\epsilon)}{(1-2\epsilon-\epsilon^2)} \, OPT $.
\end{theorem}

\begin{proof}
The number of constraints is 
\[|S||P'||\Gamma^2| + |P||\Gamma^2| + |S| = \mO(\frac{1}{\epsilon^4} \log\frac{1}{\epsilon} \log\frac{s_{max}}{s_{min}})\]

where $|S| = \mO(\frac{1}{\epsilon}\log\frac{s_{max}}{s_{min}})$ is the number of processor speeds after simplification, $|P| = \frac{1}{\epsilon} (\log\frac{s_{max}}{s_{min}} + \log\frac{1}{\epsilon})$ is the number of task sizes after simplification, $|P'| = \frac{1}{\epsilon}\log\frac{1}{\epsilon}$ is the number of different processing times after simplification, and $|\Gamma| = \frac{1}{\epsilon}$ is the number of different communication times after simplification. The number of variables is 
\[|S||\boldsymbol{C}| + |S||P||\Gamma^2| = 2^{\mathcal{O}( 1/\epsilon^2\log1/\epsilon 
+ \log\log\frac{s_{max}}{s_{min}}})\]

where $|\boldsymbol{C}| = 2^{\mathcal{O}({1/\epsilon^2\log1/\epsilon})}$ is the number of configurations.

Using the same ILP solver as before \cite{jansen2018integerlong, jansen2021eptas}, we have the following complexity for the ILP.
\[ 2^{\mathcal{O}({1/\epsilon^4\log^2 1/\epsilon \log\frac{s_{max}}{s_{min}}})} \mathcal{O}(\log N) \]

Including the binary search and pre-processing, and applying proposition \ref{procs}, for the EPAS it is 
\[2^{\mathcal{O}({1/\epsilon^4\log^2 1/\epsilon \log\frac{s_{max}}{s_{min}}})} \mathcal{O}(\log^2 N) + \mathcal{O}(N\log N)\]

The errors from the simplification steps are similar to before, except for the rounding of processor speeds which adds another $\epsilon T$ to the upper bound, giving 
\[ (1-2\epsilon-\epsilon^2)T \leq OPT \leq (1+5\epsilon)T \]
\[ \frac{(1+6\epsilon)}{(1-2\epsilon-\epsilon^2)} \, OPT \]
\end{proof}

\section{Equal Processing Costs} \label{sec:exact}
While the general problem $Q|\forkjoin, c_{ij}|\Cmax$ is sNP-hard, working with fixed computation costs can make intractable scheduling problems tractable \cite{DavidaLinton:76:A-new-algorithm-for-the-scheduling} 
or otherwise give them faster solutions, although that is not always the case \cite{Lenstra78}. 

For the following we look at scheduling on related processors with a workload of \textbf{equal} processing times.
Although this problem's complexity is still open, it is at least as hard as the version with homogeneous systems which is in turn harder than the open problem $P|r_j, p_j=p|\sum U$ in scheduling with release times and deadlines \cite{fork-join}.

\subsection{Approximation algorithm for $Q| fork-join, c_{ij}, p_j=p|\Cmax$ } \label{sec:bipartite}

We start with a guaranteed approximation algorithms that uses a simplification together with formulating the problem as finding a bipartite graph matching. Unlike with homogeneous processors \cite{fork-join}, this approach cannot solve the problem exactly, as not all schedules can be represented by a bipartite graph. A bipartite \emph{hypergraph} could fully represent this scheduling problem, but its matching problem is intractable, with other approximate solutions, hence not helpful here. 

Beginning with the simplification, we limit the resolution of start times to $p / s_\m$ on each processor $\m \in \M$, and then we have the following simple lemma \ref{lem:bipartite}, where $OPT$ is the optimal schedule length to the original instance and $T^*$ is the optimal schedule length to the simplified instance. 

\begin{lemma} \label{lem:bipartite}
$T^* - (p/ s_{min}) \leq OPT \leq T^*$
\end{lemma}
\begin{proof}
Consider an optimal schedule to the original instance. This can be transformed to a valid schedule for the simplified instance by rounding up all task start times to multiples of $p/ s_\m $ on each processor $\m \in \M$. As we round up, the schedule length of this schedule cannot be less than $OPT$, but needs to be less than the optimal schedule for the simplified instance, hence $OPT \leq T^*$

As no task can start later in this schedule by more than $p/ s_{min}$, it follows $T^* - (p/ s_{min}) \leq OPT$.

\end{proof}

In our proposed Algorithm \ref{alg:bipartite}, a bipartite graph $G=(U,V,E)$ is created for the simplified instance, where nodes $U$ represent tasks, nodes $V$ represent time slots on all processors (in $p/s_\m$ resolution), and edges $E$ connect each task to its viable time slots.

If a maximal matching for this graph includes less than $|U|=n$ edges, the original problem $P|\forkjoin, p_j=1, c_{ij}|C_{max}$ with bound $T$ is not feasible. In a maximal matching with $|U|$ edges, each edge represents an assignment of a task to a processor and start time, while no time slot is assigned to more than once, so the schedule is valid.

\begin{algorithm}
Binary search over $T$:\\
\Indp
Create graph $G(U,V,E)$ \;
\For{$\forall j \in J$}{$U \leftarrow U \cup \{u_j\} $}
\For{$\forall \m \in \M$}{
\For{$\forall \sigma \in \{ i \frac{p}{s_\m} \ | \ i \in \{ 1, 2,..., \lfloor \frac{T}{p/ s_\m} \rfloor \}\}$}{$V \leftarrow V \cup \{v_{m,\sigma}\} $}
}
\tcc{add edges}
\For{$\forall u_j \in U$}{
\For{$\forall v_{m,\sigma} \in V$}{
$E \leftarrow E \cup \{e_{j,m,\sigma}\}$ \tcc{edge from $u_j$ to $v_{m,\sigma}$}
}
}


\medskip
Find if there exists a feasible matching for the bipartite graph $G(U,V,E)$ using Ford–Fulkerson algorithm \;

\lIf{feasible matching exists}{try lower $T$}
\lElse{try higher $T$}
\Indm
\tcc{create schedule from bipartite matching}
\For{$\forall v_{j,m,\sigma} \in E$}{
schedule task $j$ to processor $m$ at time $\sigma$
}

\caption{Formulation as bipartite graph matching problem}
\label{alg:bipartite}
\end{algorithm}

\medskip
\begin{proposition} \label{prop:bipartite}
An ($OPT + \frac{p}{s_{min}}$)-schedule can be found for $P|\forkjoin, c_{ij}, p_j=p|C_{max}$ in $\mO(|\M|^2|J|^3 \log |J|)$ time.
\end{proposition}

\begin{proof}
In a maximal matching with $|U|$ edges, each edge represents an assignment of a task to a processor and start time, while no time slot is assigned to more than once, so the schedule is valid.

The reduction part in Algorithm \ref{alg:bipartite} can be done in $\mO(|J|^2)$ time, as the total number of slots (vertices in set $V$) is in the order of the total number of tasks, $\mO(|J|)$ (although the edge case of very high communication costs could add empty slots, chunks of extra slots can be optimised away), and $|U| = |J|$, so the number of edges in $E$ connecting both sets is in the order of $|J|\times|J|$. 

The maximal matching problem can be solved as a network flow problem in $\mO(V^2 U)$ time with the Ford–Fulkerson algorithm. 
To find the minimal makespan, this process has to be repeated with a binary search over $T$, i.e. $\mO(\log |J|)$ time, and repeated $|\M|^2$ times from each combination of $\m_{src}, \m_{sink}$.

With $|U|=|J|$ and $|V|=\mO(|J|)$, the total complexity is $\mO(|\M|^2|J|^3 \log |J|)$.
\end{proof}

We would like to point out that this has been better formulated than the version laid out here for homogeneous systems \cite{fork-join}, requiring only incoming communications to be simplified, which improves the bound. The complexity has also been better evaluated, from the fact that the number of slots is in the order of the number of tasks.

\subsection{Two Processors $Q2|\forkjoin, c_{ij}, p_j=p|C_{max}$} \label{sec:2proc}
Scheduling to two heterogeneous processors with equal processing costs is tractable just as it is for homogeneous processors (the problem {$P2|\forkjoin, c_{ij}, p_j=p|C_{max}$} \cite{fork-join}).

Consider the two cases separately where (case 1) the source and sink tasks are on the same processor and (case 2) where they are on separate processors. We want to reuse the same algorithms for each case (\textsc{P2-Sched1} and \textsc{P2-Sched2} \cite{fork-join}) as in the homogeneous case, but need to prove that they are still valid and obtain optimal length.

\textsc{P2-Sched1} casts the problem as a single processor throughput problem $1|r_j, p_j=p|\Sigma U_j$ on $p_2$. The tasks that are rejected then need to be executed on $p_1$ and a binary search finds the optimal solution. 
\textsc{P2-Sched1} uses an algorithm for the release time and deadline throughput problem

\textsc{P2-Sched2} is a greedy algorithm that schedules the tasks on the sink processor $p_2$ as they become available, ordered by non-increasing $\yo_i$. All tasks not fitting on $p_2$ are put on source processor $p_1$ in non-increasing $\yo_i$.
Both algorithms are used with a binary search to find the optimal.

\begin{proposition}
\textsc{P2-Sched1} \& \textsc{P2-Sched2} obtain optimal solutions for two uniformly related heterogeneous processors $Q2|\forkjoin, c_{ij}, p_j=p|C_{max}$.
\end{proposition}

\begin{proof}
\textsc{P2-Sched1} is optimal for case 1, as all tasks have the same processing time on each processor, hence the throughput argument where rejected tasks need to be scheduled on $p_1$ remains valid as in the homogeneous case, irrespective of a difference in processing speed between $p_1$ and $p_2$. 

\textsc{P2-Sched2} is optimal for case 2 as the greedy approach fills $p_2$ maximally. Only swapping tasks between $p_1$ and $p_2$ could improve the schedule. However, by the greedy approach, that cannot shorten the finish time of $p_2$, nor can it improve the outgoing communication arrival for tasks on $p_1$. Due to the related heterogeneity, all tasks have the same execution time on the respective processor. 
\end{proof}

\subsection{Unlimited Processors $Q\infty|\forkjoin, c_{ij}, p_j=p|\Cmax$} \label{sec:UnlimitedProc}
This problem is also tractable, similarly as for $P\infty|\forkjoin, c_{ij}, p_j=p|\Cmax$ \cite{fork-join}. Divide the tasks between those scheduled remotely, $J_{remote}$, and those scheduled locally, $J_{local}$, and search for a partitioning that gives the lowest schedule length. To find the schedule length given a partitioning, put each task in $J_{remote}$ on a separate processor of the faster type, and schedule $J_{local}$ using \textsc{P2-sched1/2}, potentially trying all combinations of processor types for the 2 local processors (the application of proposition \ref{procs}), repeating \textsc{P2-sched1/2} case by case, which adds $|S|^2$ to the complexity. Only the fastest $|J|$ processors need to be considered.

\subsection{Partially Equal Communications $Q|\forkjoin, p_j=p, c_{in} = c, c_{out}|\Cmax$}

Assume that outgoing communication times are arbitrary, with incoming communication times being equal (the reverse would be equivalent). This is also polynomially solvable. To briefly describe the algorithm, schedule the tasks in decreasing order of their communication times, filling the local processor first, and then after that, appending them to a processor that results in the earliest finish time at each stage. 

Again, proposition \ref{procs} can be applied here, and whether to have the source and sink together or separate, can be known by trying both. In the case that the source and sink are separate, it would be optimal to fill the sink processor first (with the most communication intensive tasks). For the remaining tasks, as all of them have equal processing time, it only matters that the most communication intensive tasks are finished earliest (the longest outgoing communications started as early as possible).

\subsection{Grouped Processors}
When we assume that communications are not incurred between processors of the same type due to close proximity or closely shared memory, this problem has similarities to scheduling with two processors, $Q2$. For the case where source and sink are on different processors types, the problem is similar enough that the same algorithm used for $Q2$ can work here. However, when source and sink are on the same processor type, with processors of the other type containing remote tasks, it again comes down to the open problem $P|p=1,r_j|\Sigma U$. A bounded approximate solution still exists by the same technique in section \ref{sec:bipartite}. 



\section{Conclusion}
In this paper we investigated the scheduling of fork-join structures, which are typical of highly parallel workloads, to heterogeneous processor systems. We presented an EPAS for scheduling to related heterogeneous processors, and some polynomial time algorithms. We see that scheduling with two types of processors (which has special relevance as CPU--accelerator systems) has largely the same results as the general case. All of these results also apply to the equivalent problems in scheduling with release times and deadlines. 

\printbibliography
\end{document}